\documentclass{amsart}
\usepackage{graphicx}
\usepackage{amsmath}
\usepackage{amsthm}
\usepackage{amssymb}
\usepackage{url}
\usepackage{algorithm,algorithmic}
\usepackage{setspace}
\usepackage{color}
\usepackage{array}
\usepackage{epsfig}
\usepackage{subfigure}

\newtheorem{thm}{Theorem}
\newcommand{\bthm}{\begin{thm}}
\newcommand{\ethm}{\end{thm}}

\newtheorem{lemma}{Lemma}
\newcommand{\blemma}{\begin{lemma}}
\newcommand{\elemma}{\end{lemma}}

\newtheorem{cor}{Corollary}
\newcommand{\bcor}{\begin{cor}}
\newcommand{\ecor}{\end{cor}}

\begin{document}

\title{Simple linear algorithms for mining graph cores}

 \author{Yang Xiang}
 \address{Department of Biomedical Informatics, The Ohio State University, Columbus, OH 43210}
  \email{yxiang@bmi.osu.edu}

\begin{abstract}
Batagelj and Zaversnik proposed a linear algorithm for the well-known $k$-core decomposition problem. However, when $k$-cores are desired for a given $k$, we find that a simple linear algorithm requiring no sorting works for mining $k$-cores. In addition, this algorithm can be extended to mine $(k_1, k_2,\ldots, k_p)$-cores from $p$-partite graphs in linear time, and this mining approach can be efficiently implemented in a distributed computing environment with a lower message complexity bound in comparison with the best known method of distributed $k$-core decomposition.
\end{abstract}


\maketitle

\section{Introduction and Basic Definitions}
Finding dense modules in graphs is of interest to many applications. In social networks, dense graph modules are often associated with communities~\cite{du2007community}. In bioinformatics, dense subgraph mining is an important approach for identifying potential biomarkers~\cite{parthasarathy2010survey} and important biomedical functions~\cite{altaf2003prediction}. However, many dense graph mining methods (e.g.~\cite{mushlin2009clique,xiang2012predicting}) are derived from clique mining or quasi-clique mining. Since enumerating all maximal cliques is an NP-hard problem~\cite{karp2010reducibility}, these dense graph mining methods usually have a high worst-case time complexity, and thus can hardly be applied to large graphs. In contrast, a $k$-core decomposition can be implemented in linear time~\cite{batagelj2003m}, and therefore it is a very popular method for graph mining and analysis.

Let $G=(V,E)$ be a graph with a vertex set $V$ and an edge set $E$. A $k$-core in a graph is a maximal connected component in which every vertex has a degree $k$ or larger. $k$-core decomposition typically refers to the identification of the maximum $k$ value for each vertex such that there exists a $k$-core containing the vertex.
To facilitate our discussion, we make the following definitions. Different from the definition of $k$-core decomposition, we define $k$-core mining to be the identification of all $k$-cores for a given $k$. We define a $k$-degree graph to be a graph in which any vertex has a degree at least $k$, and we define $G(k)$ to be the largest $k$-degree graph that is a subgraph of $G$. It is easy to conclude that $G(k)$, if exists, is unique, and any maximal connected component in $G(k)$ is a $k$-core. It is also not difficult to observe that any clique in $G$ with size $k+1$ or larger is preserved in $G(k)$. In other word, a $k$-degree graph preserves all cliques of size $k+1$ or larger, thus it can be used as a preprocessing for clique mining and related tasks.

\section{Mining $k$-cores by the linear algorithm \textsc{GraphPeel}}

Batagelj and Zaversnik~\cite{batagelj2003m} proposed a linear implementation of the $k$-core decomposition. However, in some applications $k$-cores are desired for a given $k$. In this case, we find that an alternative linear algorithm requiring no sorting works for mining $k$-cores. As a result, this linear algorithm is much simpler to implement. To the best of our knowledge, this discovery has not been introduced in literature. Algorithm~\ref{alg:graphPeelOff}, ~\textsc{GraphPeel}, is the pseudocode of the simple $k$-core mining algorithm. It scans all the vertices only once (the for loop at Step 3), and start a recursive processing on a vertex with degree less than $k$ (the while loop at Step 7).

\begin{algorithm}
\begin{scriptsize}
\caption{\textsc{GraphPeel($G=(V,E)$, $k$)}}
\label{alg:graphPeelOff}
\begin{algorithmic}[1]
\STATE Initialize $activeVertex$ vector to be all true;
\STATE Initialize $Counter$ vector to be the degree of each vertex;
\FORALL{$v\in V$}
    \IF{$activeVertex(v)==true$ \&\& $Counter(v)<k$}
        \STATE $activeVertex(v)=false$;
        \STATE Add $v$ into $Q$;
        \WHILE{$Q\neq\emptyset$ }
            \STATE $w=dequeue(Q)$;
            \FORALL{$u$ adjacent to $w$}
                \STATE $Counter(u)=Counter(u)-1$;
                \IF{$activeVertex(u)==true$ \&\& $Counter(u)<k$}
                    \STATE $activeVertex(u)=false$;
                    \STATE enqueue $u$ onto $Q$;
                \ENDIF
            \ENDFOR
        \ENDWHILE
    \ENDIF
\ENDFOR
\STATE Regenerate $G$ induced by active vertices.
\end{algorithmic}
\end{scriptsize}
\end{algorithm}

In the following, we prove Theorems~\ref{thm:correctness} and~\ref{thm:time}, which state the correctness and linear time complexity of \textsc{GraphPeel} algorithm for mining $k$-cores.

\begin{thm}\label{thm:correctness}
Algorithm~\ref{alg:graphPeelOff} \textsc{GraphPeel} processes the graph $G$ into $G(k)$.
\end{thm}
\begin{proof}
Let $G'$ be the result graph after applying Algorithm~\ref{alg:graphPeelOff} on $G$. We need to prove that $G'=G(k)$. The definition of $G(k)$ implies that $G'=G(k)$ if and only if both of the following two statements hold:
\begin{itemize}

\item[1.] Any vertex in $G'$ has a degree at least $k$.
\item[2.] $G(k)$ is a subgraph of $G'$.
\end{itemize}
Proof of Statement 1:\\
Let us assume that there exists a vertex $v$ in $G'$ such that $d_{G'}(v)<k$. We will show in the following this is a contradiction. The degree of this vertex in $G$ is either of the two cases: (1) $d_G(v)<k$; (2) $d_G(v)\geq k$.

Case (1) is not possible because if $activeVertex(v)==true$ when the for loop (Step 3) reaches $v$, $activeVertex(v)$ will be turned into $false$. Therefore $activeVertex(v)$ must be false before $G'$ is generated, thus $v$ will not be included in $G'$.

Case (2) implies that the degree of $v$ is reduced below $k$ during the execution of Algorithm~\ref{alg:graphPeelOff}. However, the only degree reduction in Algorithm~\ref{alg:graphPeelOff} takes place at Step 10, and it is immediately followed by degree check (step 11 and 12). That is, immediately after the degree of $v$ is reduced below $k$, $activeVertex(v)$ will be turned into false if it is true. In another word, $v$ will not be included in $G'$.

Since both cases are not possible, we have reached a contradiction and proved Statement 1.

Proof of Statement 2:\\
Let us assume that there exists a vertex $v$ in $G(k)$ such that $v\not\in G'$. We will show in the following this is a contradiction. Since $v\not\in G'$, $v$ was removed (i.e., $activeVertex(v)$ marked as false) during the execution of Algorithm~\ref{alg:graphPeelOff}. Since $v\in G(k)$, we conclude that $activeVertex(v)$ can only be turned into false at Step 12, which was caused by the removal of $d_G(v)-k+1$ vertices from $v$'s neighborhood. Once again, since $v\in G(k)$, we conclude that at least one vertex (letting it be $w$) among the $d_G(v)-k+1$ removed vertices belongs to $G(k)$. Since $w$ was removed, we conclude $w\not\in G'$. We can apply this reasoning on $w$ and recursively repeat this process until we come to $u$, the first vertex in $G(k)$ which was removed by Algorithm~\ref{alg:graphPeelOff}. Then we reach a contradiction because $Counter(u)\geq k$ always hold before $u$ was removed.
\end{proof}

\begin{thm}\label{thm:time}
Algorithm~\ref{alg:graphPeelOff} \textsc{GraphPeel} runs in $O(|V|+|E|)$ time.
\end{thm}
\begin{proof}
First, it is easy to see each vertex will be enqueued and dequeued at most once, the active status of each vertex will be changed at most once, and status check (Step 4) immediately follow the for-loop takes at most once on each vertex. More importantly, observing that a visit to an adjacent vertex (Step 9) is an edge visit, we conclude that Algorithm~\ref{alg:graphPeelOff} \textsc{GraphPeel} visit each edge at most once (i.e., only when one end vertex is dequeued). This edge visit includes updating the $Counter$ of an end vertex (Step 10) and checking its status (Step 11), both taking constant time. As a conclusion, Algorithm~\ref{alg:graphPeelOff} \textsc{GraphPeel} takes $O(|V|+|E|)$ time in the worst case.
\end{proof}
As we can see, the \textsc{GraphPeel} algorithm for mining $k$-cores is succinct and easy to implement. More importantly, this approach can be applied to mine $(k_1, k_2,\ldots,k_p)$-cores and extended to a distributed environment resulting in a low message complexity bound, as we will describe in the following two sections.

\section{Mining$(k_1, k_2,\ldots,k_p)$-cores from $p$-partite graphs}
A $p$-partite graph is defined as a graph whose vertices can be partitioned into $p$ disjoint sets such that no two vertices within the same set are adjacent. $p$-partite graphs ($p\geq 2$) have been frequently used to model real data and are of interest to many applications. For example, identifying dense components in bipartite graphs ($p=2$) is an interesting data mining problem which can be associated with mining frequent itemsets~\cite{xiang2011summarizing}. Similar to mining $k$-cores, a corresponding $(k_1, k_2,\ldots,k_p)$-core mining algorithm is available for $p$-partite graphs. Such an algorithm can serve as efficient preprocessing for many applications on a $p$-partite graph. As an example of $(k_1, k_2,\ldots,k_p)$-core application, $(k_1, k_2)$-cores from bipartite graphs ($p=2$ in this case) have been used for network visualization~\cite{ahmed2007visualisation}.

Let $G=(V_1, V_2,\ldots,V_p, E)$ be a $p$-partite graph. A $(k_1, k_2,\ldots,k_p)$-core is a maximal connected component in $G$ such that each vertex in $V_i$ has a degree $k_i$ or larger, for any $i$ between 1 and $p$. We define a $(k_1, k_2,\ldots,k_p)$-degree graph to be a $p$-partite graph in which any vertex in $V_i$ has a degree at least $k_i$ (note that $V_i$ is allowed to be an empty set). We define $G(k_1, k_2,\ldots,k_p)$ to be the largest $(k_1, k_2,\ldots,k_p)$-degree graph that is a subgraph of $G$. Again, it is easy to conclude that $G(k_1, k_2,\ldots,k_p)$, if exists, is unique.

The graph peel approach introduced in the previous section can be extended to perform the mining of $(k_1, k_2,\ldots,k_p)$-cores from $p$-partite graphs. The pseudocode is listed in Algorithm~\ref{alg:bipartiteGraphPeelOff}. By literally following the proofs for Theorems~\ref{thm:correctness} and~\ref{thm:time}, we can show that Algorithm~\ref{alg:bipartiteGraphPeelOff} correctly performs the mining of $(k_1, k_2,\ldots,k_p)$-cores for the given $p$-partite graph, and it runs in linear time, as state in Corollaries~\ref{cor:correctness} and~\ref{cor:linear}.
\begin{cor}\label{cor:correctness}
Algorithm~\ref{alg:bipartiteGraphPeelOff} \textsc{p-partiteGraphPeel} processes the $p$-partite graph $G$ into $G(k_1, k_2,\ldots,k_p)$.
\end{cor}
\begin{cor}\label{cor:linear}
Algorithm~\ref{alg:bipartiteGraphPeelOff} \textsc{p-partiteGraphPeel} runs in $O(\sum_{i=1}^{p}|V_i|+|E|)$ time.
\end{cor}

\begin{algorithm}
\begin{scriptsize}
\caption{\textsc{p-partiteGraphPeel}($G=(V_1, V_2,\ldots,V_p, E)$, $(k_1, k_2,\ldots,k_p)$)}
\label{alg:bipartiteGraphPeelOff}
\begin{algorithmic}[1]
\STATE Initialize $activeVertex$ vector to be all true;
\STATE Initialize $Counter$ vector to be the degree of each vertex;
\FOR{$i=1$ to $p$}
    \FORALL{$v\in V_i$}
        \IF{$activeVertex(v)==true$ \&\& $Counter(v)<k_i$}
            \STATE $activeVertex(v)=false$; Add $v$ into $Q$; IterativePeel($Q$);
        \ENDIF
    \ENDFOR
\ENDFOR
\STATE Regenerate $G$ induced by active vertices.
\end{algorithmic}
\begin{algorithmic}
\STATE PROCEDURE IterativePeel($Q$):
\WHILE{$Q\neq\emptyset$}
    \STATE $w=dequeue(Q)$;
    \FORALL{$u$ adjacent to $w$}
        \STATE $Counter(u)=Counter(u)-1$;
        \IF{$activeVertex(u)==true$ \&\& $Counter(u)<k_{p(u)}$\COMMENT{assuming $u\in V_{p(u)}$}}
            \STATE $activeVertex(u)=false$; enqueue $u$ onto $Q$;
        \ENDIF
    \ENDFOR
\ENDWHILE
\end{algorithmic}
\end{scriptsize}
\end{algorithm}

\section{Distributed core mining}
Another big advantage of the graph peel approach is that it can be easily implemented in a distributed environment. The implementation is much succinct than~\cite{montresor2011distributed}, the best known method available for distributed k-core decomposition. Furthermore, this simple distributed $k$-core mining algorithm has a lower message complexity and comparable time complexity compared to~\cite{montresor2011distributed} for the purpose of identifying all $k$-cores of a given $k$. The distributed implementation of Algorithms~\ref{alg:graphPeelOff} and ~\ref{alg:bipartiteGraphPeelOff} can be achieved with similar workflows. For succinctness, we combine the two efforts into one set of pseudocodes: Algorithms~\ref{alg:distributedGraphPeelOffInit} and~\ref{alg:distributedGraphPeelOffMessage}.

\begin{algorithm}
\begin{scriptsize}
\caption{\textsc{onInitial()}}
\label{alg:distributedGraphPeelOffInit}
\begin{algorithmic}[1]
\IF{$degree(v)<k(v)$\COMMENT{for distributed version of Algorithm~\ref{alg:graphPeelOff}, $k(v)=k$; for distributed version of Algorithm~\ref{alg:bipartiteGraphPeelOff}, $k(v)=k_{p(v)}$ where $v\in V_{p(v)}$}}
    \FOR{$u\in neighbors$}
        \STATE $SendOffMessage(u)$;
    \ENDFOR
    \STATE $Node\_Status=off$;\COMMENT{This node becomes inactive.}
\ENDIF
\end{algorithmic}
\end{scriptsize}
\end{algorithm}

\begin{algorithm}
\begin{scriptsize}
\caption{\textsc{onMessage()}}
\label{alg:distributedGraphPeelOffMessage}
\begin{algorithmic}[1]
\STATE $degree=degree-1$;
\IF{$degree<k(v)$\COMMENT{for distributed version of Algorithm~\ref{alg:graphPeelOff}, $k(v)=k$; for distributed version of Algorithm~\ref{alg:bipartiteGraphPeelOff}, $k(v)=k_{p(v)}$ where $v\in V_{p(v)}$}}
    \FOR{$u\in neighbors$}
        \STATE $SendOffMessage(u)$;
    \ENDFOR
    \STATE $Node\_Status=off$;\COMMENT{This node becomes inactive.}
\ENDIF
\end{algorithmic}
\end{scriptsize}
\end{algorithm}

All nodes start with Algorithm~\ref{alg:distributedGraphPeelOffInit} \textsc{onInitial()}. If a node degree is less than $k(v)$, it sends ``off'' messages to neighbors and then becomes inactive, otherwise, it goes into the suspension status and can be awaken to run Algorithm~\ref{alg:distributedGraphPeelOffMessage} \textsc{onMessage()} upon the arrival of a new message. To simplify our discussion, we assume once a node becomes inactive, it will keep dormant and unresponsive to any coming messages.

By literally following the proof of Theorem~\ref{thm:correctness}, it is easy to show that Algorithms~\ref{alg:distributedGraphPeelOffInit} and~\ref{alg:distributedGraphPeelOffMessage} correctly perform the core mining, as stated in the following corollaries.
\begin{cor}
The distributed $k$-core mining with Algorithms~\ref{alg:distributedGraphPeelOffInit} and~\ref{alg:distributedGraphPeelOffMessage} applied to each node, processes the network $G=(V,E)$ into $G(k)$.
\end{cor}
\begin{cor}
The distributed $(k_1,k_2,\ldots,k_p)$-core mining, with Algorithms~\ref{alg:distributedGraphPeelOffInit} and~\ref{alg:distributedGraphPeelOffMessage} applied to each node, processes the network $G=(V_1,V_2,\ldots,V_p,E)$ into $G(k_1,k_2,\ldots,k_p)$.
\end{cor}

The distributed core mining by Algorithms~\ref{alg:distributedGraphPeelOffInit} and~\ref{alg:distributedGraphPeelOffMessage} also has a low message complexity, as stated in Lemma~\ref{lemma:distributedCorrect}.
\begin{lemma}\label{lemma:distributedCorrect}
The message complexity of the distributed core mining is bounded by $O(|E|)$.
\end{lemma}
\begin{proof}
For any edge connecting two vertices, it will be used at most once by each end vertex to send the off-message. Thus, the lemma is correct because at most two off-messages will pass one edge.
\end{proof}

In addition, the distributed core mining by Algorithms~\ref{alg:distributedGraphPeelOffInit} and~\ref{alg:distributedGraphPeelOffMessage} converge fast under certain circumstances as discussed in the following. To facilitate our analysis, we assume each node works under synchronized phases. In each phase, a node is able to receive all messages sent out by other nodes in the previous phase, and if applicable sent off-messages to all neighbors. Then we have the following lemma.
\begin{lemma}
The distributed core mining, with Algorithms~\ref{alg:distributedGraphPeelOffInit} and~\ref{alg:distributedGraphPeelOffMessage} applied to each node, converges in no more than $|V|-k$ phases for a $k$-core mining on Graph $G=(V,E)$ where $k\leq |V|$, or no more than $(|V_1|-k_2)+(|V_2|-k_1)$+1 phases for a $(k_1,k_2)$-core mining on Bipartite Graph $G=(V_1, V_2$), or no more than $\sum_{i=1}^{p}(|V_i|)-\min_{1\leq i\leq p}(k_i)$ phases for a $(k_1,k_2,\ldots,k_p)$-core mining on Graph $G=(V_1, V_2, \ldots, V_p, E)$ where $k_i\leq |V_i|$ for any $1\leq i\leq p$.
\end{lemma}
\begin{proof}
Let us consider the case of mining $k$-cores which takes at least 2 phases to finish (the lemma is clearly true if a $k$-core mining finishes in 1 phase). It is easy to prove by contradiction that for any phase before the last phase, there are at least $k+1$ active nodes at the beginning of the phase. Given this, we only need to show that before the last phase, at least one node will be set off in each phase. This can be proved by contradiction. If at one phase $p$ before the last phase no node is set off, then at the next phase no message will reach any node, thus again no node will be set off and we reach a contradiction that the distributed $k$-core mining completes on or before phase $p$.\\
The proofs for the case of a $(k_1,k_2)$-core mining and a $(k_1,k_2,\ldots,k_p)$-core mining can be obtained by following the above strategy.
\end{proof}

\section{Conclusion}
We have shown that a simple linear algorithm works for mining $k$-cores, i.e., identifying all $k$-cores for a given $k$. Based on this algorithm, we have shown that mining $(k_1, k_2,\ldots,k_p)$-cores on $p$-partite graphs can be implemented in a simple linear approach. In addition, these simple linear algorithms can be extended to a distributed environment with a low message complexity and a short running time. We expect core mining algorithms described in this manuscript will contribute to large data processing for many real applications.

\label{}

\bibliographystyle{elsarticle-num}
\begin{small}
\begin{singlespace}
\bibliography{core}

\begin{thebibliography}{10}
\expandafter\ifx\csname url\endcsname\relax
  \def\url#1{\texttt{#1}}\fi
\expandafter\ifx\csname urlprefix\endcsname\relax\def\urlprefix{URL }\fi
\expandafter\ifx\csname href\endcsname\relax
  \def\href#1#2{#2} \def\path#1{#1}\fi

\bibitem{du2007community}
N.~Du, B.~Wu, X.~Pei, B.~Wang, L.~Xu, Community detection in large-scale social
  networks, in: Proceedings of the 9th WebKDD and 1st SNA-KDD 2007 workshop on
  Web mining and social network analysis, ACM, 2007, pp. 16--25.

\bibitem{parthasarathy2010survey}
S.~Parthasarathy, S.~Tatikonda, D.~Ucar, A survey of graph mining techniques
  for biological datasets, Managing and Mining Graph Data (2010) 547--580.

\bibitem{altaf2003prediction}
M.~Altaf-Ul-Amin, K.~Nishikata, T.~Koma, T.~Miyasato, Y.~Shinbo,
  M.~Arifuzzaman, C.~Wada, M.~Maeda, T.~Oshima, H.~Mori, et~al., Prediction of
  protein functions based on k-cores of protein-protein interaction networks
  and amino acid sequences, GENOME INFORMATICS SERIES (2003) 498--499.

\bibitem{mushlin2009clique}
R.~Mushlin, S.~Gallagher, A.~Kershenbaum, T.~Rebbeck, Clique-finding for
  heterogeneity and multidimensionality in biomarker epidemiology research: The
  chamber algorithm, PloS one 4~(3) (2009) e4862.

\bibitem{xiang2012predicting}
Y.~Xiang, C.~Zhang, K.~Huang, Predicting glioblastoma prognosis networks using
  weighted gene co-expression network analysis on tcga data, BMC Bioinformatics
  13~(Suppl 2) (2012) S12.

\bibitem{karp2010reducibility}
R.~Karp, Reducibility among combinatorial problems, 50 Years of Integer
  Programming 1958-2008 (2010) 219--241.

\bibitem{batagelj2003m}
V.~Batagelj, M.~Zaversnik, An o (m) algorithm for cores decomposition of
  networks, Arxiv preprint cs/0310049, 2003 (a journal version entitled "Fast
  algorithms for determining (generalized) core groups in social networks"
  appeared in Advances in Data Analysis and Classification, 2011. Volume 5,
  Number 2, 129-145).

\bibitem{montresor2011distributed}
A.~Montresor, F.~De~Pellegrini, D.~Miorandi, Distributed k-core decomposition,
  IEEE Transactions on Parallel and Distributed Systems 24~(2) (2013) 288--300.

\bibitem{xiang2011summarizing}
Y.~Xiang, R.~Jin, D.~Fuhry, F.~Dragan, Summarizing transactional databases with
  overlapped hyperrectangles, Data Mining and Knowledge Discovery 23~(2) (2011)
  215--251.

\bibitem{ahmed2007visualisation}
A.~Ahmed, V.~Batagelj, X.~Fu, S.-H. Hong, D.~Merrick, A.~Mrvar, Visualisation
  and analysis of the internet movie database, in: Visualization, 2007.
  APVIS'07. 2007 6th International Asia-Pacific Symposium on, IEEE, 2007, pp.
  17--24.

\end{thebibliography}
\end{singlespace}
\end{small}
\end{document}